\documentclass[sigplan,9pt]{acmart}
\pdfmapfile{+txfonts.map}
\startPage{1}

\usepackage{booktabs}   
\usepackage{subcaption} 
\usepackage{tikz,braket,qcircuit,thm-restate,bm,cuted}
\usetikzlibrary{arrows,positioning,matrix,backgrounds,calc,fit}

\newcommand{\cc}{\mathcal{C}}

\newcommand{\co}{\mathcal{O}}

\newcommand{\cm}{\mathcal{M}}

\newcommand{\Z}{\mathbb{Z}}
\newcommand{\I}{\mathbb{I}}
\newcommand{\om}{\omega}
\newcommand{\empmod}{P}

\declaretheorem[name=Theorem]{thm}
\declaretheorem[name=Definition]{defn}

\begin{document}

\title{Logical paradoxes in quantum computation}         
                                        
\author{Nadish de Silva}

\affiliation{
  \department{Department of Computer Science}              
  \institution{University College London}            
  \streetaddress{Street1 Address1}
  \city{London}
  \postcode{Post-Code1}
  \country{UK}                    
}
\email{nadish.desilva@utoronto.ca}          

\begin{abstract}
The precise features of quantum theory enabling quantum computational power are unclear.  Contextuality---the denial of a notion of classical \emph{physical reality}---has emerged as a promising hypothesis: e.g. Howard \emph{et al.} showed that the magic states needed to practically achieve quantum computation are contextual.  

Strong contextuality, as defined by Abramsky-Brandenburger, is an extremal form of contextuality describing systems that exhibit logically paradoxical behaviour.

After introducing number-theoretic techniques for constructing exotic quantum paradoxes, we present large families of strongly contextual magic states that are computationally optimal in the sense that they enable universal quantum computation via deterministic injection of gates of the Clifford hierarchy. We thereby bolster a refinement of the resource theory of contextuality that emphasises the computational power of logical paradoxes.
\end{abstract}

%



\copyrightyear{2018} 
\acmYear{2018} 
\setcopyright{acmlicensed}
\acmConference[LICS '18]{LICS '18: 33rd Annual ACM/IEEE Symposium on Logic in Computer Science}{July 9--12, 2018}{Oxford, United Kingdom}
\acmBooktitle{LICS '18: LICS '18: 33rd Annual ACM/IEEE Symposium on Logic in Computer Science, July 9--12, 2018, Oxford, United Kingdom}
\acmPrice{15.00}
\acmDOI{10.1145/3209108.3209123}
\acmISBN{978-1-4503-5583-4/18/07}

\maketitle

\section{Introduction}

\textbf{Issues.  }Identification of the precise features of quantum theory accounting for quantum advantages over classical devices, and of the mechanisms by which they do so, remains an open problem in quantum computation.  Physicists and quantum information theorists seek high-level principles behind quantum advantages in order to optimise resources and serve as guideposts towards novel protocols and algorithms.

A promising, emerging hypothesis is \emph{contextuality}: a concept from the foundations of quantum mechanics first articulated by Bell-Kochen-Specker \cite{bell1964einstein, ks} that is commonly understood as the denial of a classical notion of reality.  Essentially, contextuality is the failure of correlated probabilistic data to be reproducible by a classical probability model.  It subsumes \emph{nonlocality}---the failure of separated systems to respect a strong classical assumption about causality \cite{bell66}---as a special case.  With these concepts, tensions between classical and quantum physics, e.g. the Einstein-Podolsky-Rosen paradox \cite{EPR}, are formally articulated to a degree that experiments (e.g. \cite{aspect,hensen2015loophole}) can be performed to refute classical ontological assumptions.

An important lesson of recent studies of contextuality is that it is a logical phenomenon that is most clearly understood independently of quantum physics, the setting in which it first appeared \cite{AB,abramsky2017contextual,posspoly,tqc, logicalbi,ccp,quantummonad,atserias2017generalized,Kishida16}.  At its core, contextuality is concerned with systems (quantum or otherwise) possessing properties whose values are not assumed to be simultaneously ascertainable.  The primary question asked of such systems is whether they can, in principle, be modelled under the assumption that all properties do simultaneously have consistent values.  The central technical tools are Bell inequalities \cite{bell66, klyachko, csw} that bound the strength of classical correlations.

The notion of \emph{strong contextuality}, due to Abramsky-Brandenburger \cite{AB}, is an extremal form of contextuality describing probabilistic systems that exhibit logically paradoxical behaviour.  Whereas standard contextuality is witnessed by the violation of a probabilistic Bell inequality, strong contextuality is witnessed by a logical paradox \cite{logicalbi}.  These paradoxes lie at the heart of contextuality in the sense that any contextual data is a convex mixture of classical and strongly contextual data \cite{abramsky2017contextual}.

Access to the contextual correlations that can arise from quantum systems has been suggested as a critical resource for realising the advantages of quantum devices for a variety of computational \cite{anders, raussendorf, delfosse2015wigner, veitch, pashayan2015estimating,grudka,memory,galvao2005discrete} and communicational \cite{vandamthesis,cubitt,cleve2004consequences,abramsky2017contextual} tasks.  Our motivation is to consider the role of logical paradoxes realised by extremally contextual quantum resources for the task of practically achieving universal quantum computation, i.e. the ability to perform any quantum gate with arbitrary precision.

  We build on recent, seminal work of Howard \emph{et al.} \cite{howard}, following \cite{wootters1986discrete, gibbons2004discrete, galvao2005discrete, cormick2006classicality, gross2006hudson, veitch, mari2012positive, veitch2}, establishing the neccessity of probabilistically contextual resources for achieving universal quantum computation in the model of fault-tolerant quantum computation via state injection.  In this model, a heavily restricted quantum computer capable of only \emph{stabilizer operations} \cite{gotthesis} (whose circuits use only Clifford gates) is implemented.  These restrictions facilitate the error-correction necessary of any practical model of quantum computing; however, they are so strong that the scheme is efficiently classically simulable \cite{gk} and thus offers none of the power of quantum computation.  The ability to perform a non-Clifford gate promotes this scheme to universal quantum computation.  This can be achieved by state injection \cite{gottesman1999demonstrating}: a variation on the quantum teleportation protocol that consumes a resource \emph{magic state} and implements a non-Clifford gate.  In essence, states can be converted, via stabilizer operations, into gates.  Typical magic states must undergo a stochastic distillation and injection process, due to Bravyi and Kitaev \cite{bravkit}, requiring randomly many repeated attempts before being successfully converted into a non-Clifford gate.  Howard \emph{et al.} \cite{howard} showed that single-qudit magic states of odd, prime dimension exhibit standard contextuality with respect to two-qudit stabilizer measurements once paired with an ancilla qudit.

Logical paradoxes are known to be critical resources in a diverse variety of  settings: e.g. nonlocal games \cite{cleve2004consequences, briet2013multipartite}, zero-error information theory \cite{cubitt}, measurement-based  computation \cite{raussendorf2001one}, etc.  Across these settings, a striking pattern emerges: access to  contextual resources can improve the probability of success of achieving some task whereas access to paradoxes can enable deterministic advantage.  We are interested in whether this pattern extends to the complex task of universal quantum computation.

\textbf{Results.  }In this paper, we present large families of multiqudit magic state-based paradoxes (Theorem 1) arising from the Clifford hierarchy \cite{gotthesis}, a distinguished class of gates admitting direct, deterministic injection.  
They are fundamentally different from all known quantum paradoxes in requiring novel number-theoretic arguments (Lemmas 3.1 and 3.2) that we introduce below.  

\textbf{Significance.  }  The existence of these paradoxes is striking given that there is no a priori reason that they should be realised by  quantum states, let alone by those magic states that serve as ideal computational resources. Finding these complex paradoxes is possible only by testing the hypothesis of paradox-driven advantage.  The alignment of resourcefulness of magic states with extremal contextuality in the the logical sense validates logical aspects of contextuality as highly useful from a computational perspective.  Extending the theme of probabilistic vs. deterministic advantages from standard vs. strong contextuality from simpler tasks (e.g. measurement-based computation of nonlinear finite functions) to the complex setting of full-blown universal quantum computation suggests a common underlying logical mechanism behind contextuality-powered advantages in computation and communication.  This suggests further scope for the use of logical methods in reaching a structural understanding of quantum computational advantage.

\textbf{Contents.  }In the next sections, we briefly review contextuality, contextuality as a computational resource, and the magic state model of quantum computation before presenting our paradoxical magic states and indicating future directions towards understanding the paradoxes at the heart of contextuality as powerful computational resources.

\section{Background}
\subsection{Contextuality: probabilistic and paradoxical}
Here, we briefly summarize the prerequisite background on contextuality, following the framework of Abramsky-Brandenburger \cite{AB}.   We emphasize that contextuality is a property of probabilistic data with respect to a set of measurements in an experimental setup wherein it is not assumed that all measurements can be performed together.  Thus, it is a concept that applies not to any particular physical theory, but to the empirical data predicted by a physical theory.  Of course, in the sequel, we are concerned with contextuality of data that arises from performing quantum measurements on quantum states.

An experiment is modelled by a \emph{measurement scenario}: a triple $(\cm, \cc, \co)$ where $\cm$ is a set of measurements and $\cc$ is the set of \emph{contexts}.  A context is a maximal set $C \subset \cm$ of measurements that can be performed together.  The outcome set $\co$ is merely a set of labels for the outcomes possible upon performing a measurement.  Nonlocality experiments are those where each measurement is associated to a site and contexts are choices of a measurement from each site.  For example, the standard Bell scenario  \cite{bell1964einstein} is captured by $\cm = \{A_0, A_1, B_0, B_1\}$, contexts $\cc = \{\{A_a,B_b\}: a,b \in \{0,1\}\}$, and outcome set $\co = \{0,1\}$.  Measurement of each $M \in \cm$ yields a value from the outcome set $\co$; a run of the experiment for the context $C$ produces a joint outcome $o: C \to \co$.   Empirical data from performing experiments on a system in a fixed state is captured by a choice, for each context $C$, of a conditional distribution $\empmod_C = P(-|C): \co^C \to [0,1]$ on joint outcomes.  Data satisfying physically reasonable conditions of generalised nonsignalling---the marginal distributions $\empmod_C|_{C \cap C'}$ and $\empmod_{C'}|_{C \cap C'}$ agree for all $C$ and $C'$---constitute an \emph{empirical model} $\empmod$.

An empirical model $\empmod$ is \emph{noncontextual} (\emph{local}) when its predictions can be accounted for by a \emph{noncontextual (locally causal) hidden variable model}.  Such a hidden variable model can be assumed to have a canonical form \cite{fine,AB}: the hidden variable space is $\Lambda = \co^\cm$, with each hidden variable being a function $\lambda: \cm \to \co$ (a choice of predetermined outcome for all measurements), together with a distribution $\mu: \Lambda \to [0,1]$ over hidden variables that recovers the $\empmod_C$ as marginal distributions $\mu|_C$.  Data arising from a hidden variable model will satisfy all Bell inequalities \cite{peres1999all, klyachko, logicalbi, csw} that bound the strength of classical correlations; thus, contextuality is witnessed by violation of a probabilistic inequality.

A hidden variable $\lambda: \cm \to \co$ and an empirical model $\empmod$ are \emph{consistent} when, for each context $C$, the joint outcome $\lambda$ prescribes to the measurements in $C$ is \emph{possible}:  $\empmod_C(\lambda|_C) \neq 0$.  An empirical model is \emph{strongly contextual} when it is inconsistent with all hidden variables, i.e. for every hidden variable $\lambda$, there is an experiment for which $\lambda$ predicts an impossible outcome.  Whereas the standard probabilistic form of contextuality is witnessed by a violation of an inequality, strong contextuality is witnessed by a logical paradox \cite{ccp}.  Strong contextuality generalises the notion of \emph{maximal nonlocality} of Elitzur-Popescu-Rohrlich \cite{eprohr} and Barrett \emph{et al.} \cite{barrettmax}.  Well-known examples include the GHZ state \cite{ghz} (see Figure \ref{merminstar}) and the PR box \cite{popescu1994quantum}.

\begin{figure}[h]
  \begin{center}
    \includegraphics[scale=0.12]{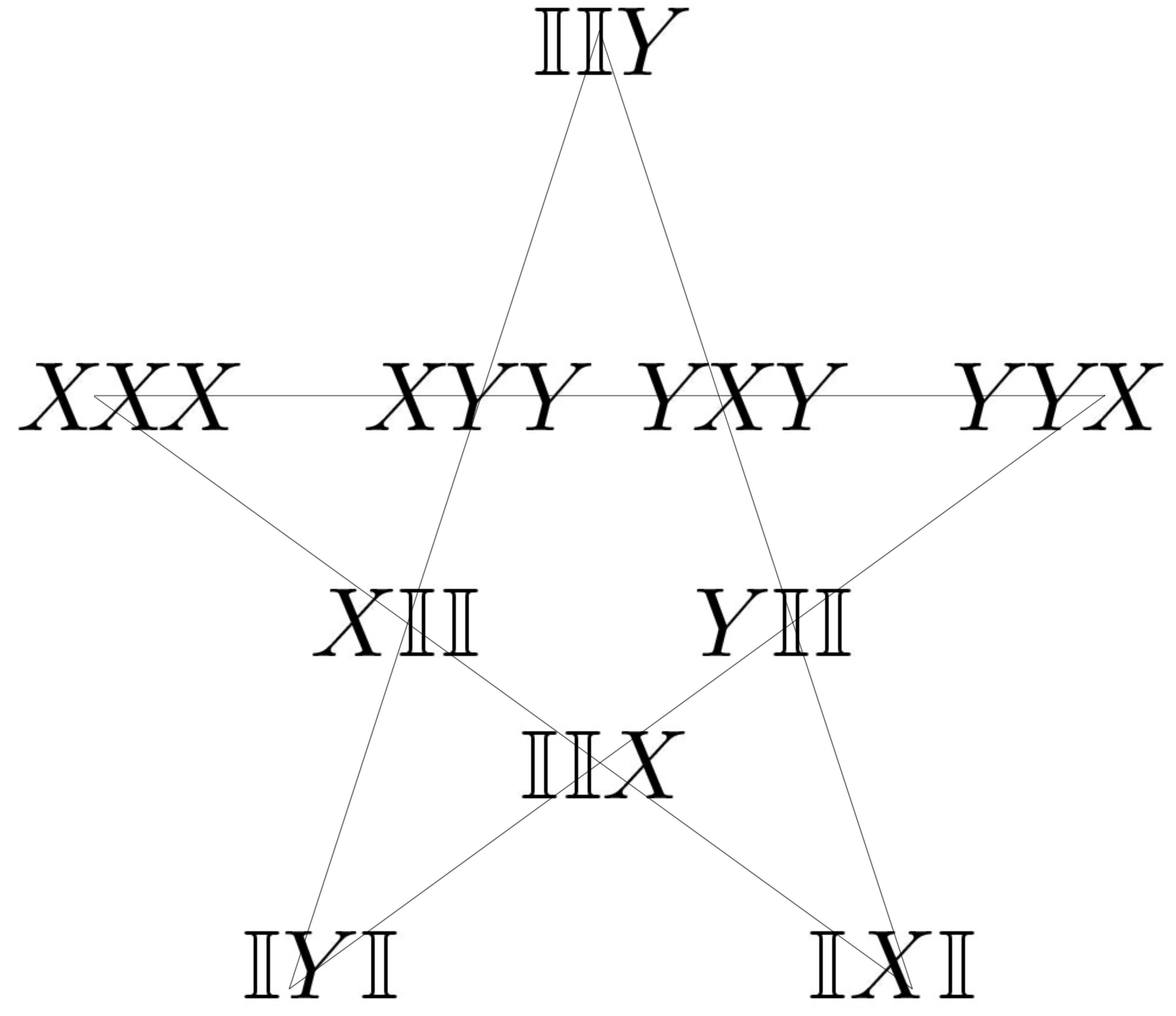}
  \end{center}
  \caption{Mermin's star paradox \cite{mermin}.  The GHZ state is an eigenstate of the qubit Pauli measurements on the horizontal line with eigenvalues +1,-1,-1,-1 (from right to left) thereby fixing the value of any consistent hidden variable $\lambda$ on these measurements.  However, no assignment of outcomes to the other measurements can respect the parity consistency constraints imposed by the algebraic relations between measurements: since the product of the measurements along any non-horizontal line is the identity, the product of outcomes along these lines must be 1.  Thus, the GHZ state is strongly contextual with respect to the above measurements.  This linear algebraic paradox can be cast as a constraint satisfaction problem \cite{cleve2014characterization}.   The magic states we will consider below need not be eigenstates of any stabiliser measurements.}
  \label{merminstar}
\end{figure}

The intuitive idea that a strongly contextual system behaves paradoxically can be formalised.  For each measurement $M \in \cm$ and outcome $o \in \co$, the symbol $M \to o$ is interpreted as ``$M$ is measured resulting in the outcome $o$".  One can construct a \emph{sentence} with, for $M \in C$ of a fixed context and $o \in \co$, symbols $M \to o$ and the connectives AND, OR, and NOT.  The \emph{theory} of an empirical model are all those sentences that are always true after any run of the experiment.  An empirical model is strongly contextual if and only if its theory is logically inconsistent.

Every empirical model $\empmod$ is a convex mixture of a noncontextual part $\empmod^{NC}$ and a strongly contextual part $\empmod^{SC}$ as $\empmod = \text{CF}(\empmod)\empmod^{SC} + [1 - \text{CF}(\empmod)] \empmod^{NC}$.  
Here, $\text{CF}(\empmod)$ is a measure of contextuality known as the \emph{contextual fraction} \cite{AB}; it generalises the nonlocal fraction \cite{barrettmax}.
An empirical model is strongly contextual if and only if $\text{CF}(\empmod)$ is 1 (or, the \emph{noncontextual fraction} $[1-\text{CF}(\empmod)]$ is 0).  Thus, the contextual fraction measures the degree to which a model provides probabilistic access to paradoxical data.  Strongly contextual models are extremally contextual in a geometric sense: they correspond to points on those faces of the nonsignalling polytope with no noncontextual vertices \cite{posspoly, acin}.

\subsection{Contextuality as a computational resource}

Contextuality and nonlocality have been shown to be critical to achieving advantages in a variety of informatic tasks.  A general scheme that results of this flavour fit is as follows: a restricted model of computation or communication is proven to have some upper bound on its success in achieving some task before showing that same model, once supplemented with the ability to perform measurements on a contextual system, gains a probabilistic advantage.   The pattern of strongly contextual resources conferring deterministic advantages is seen across computational and communicational settings.  
Sharing an unlimited number of PR boxes between two parties renders all communication complexity problems trivial \cite{vandamthesis}.  Kochen-Specker configurations play an essential role in boosting the zero-error (i.e. deterministic) capacity of certain classical channels \cite{cubitt}.  Perfect strategies for nonlocal games \cite{cleve2004consequences} require strongly contextual resources \cite{abramsky2017contextual}.

We are primarily interested in two computational tasks that are known to benefit from access to contextuality: measurement-based computation and fault-tolerant quantum computation via state injection.  Measurement-based computation is a generalisation of the model of measurement-based quantum computation (MBQC).  In MBQC, a restricted classical computer, having received classical input, directs measurements on an entangled resource quantum state.  Remarkably, although no quantum gates are directly implemented, such models are capable of universal quantum computation.  Measurement-based computation is an abstraction of MBQC wherein no particular background physical theory is assumed.  Instead, measurements are performed on an abstract resource system whose behaviour is described by an empirical model.  Below, we describe known results on the necessity of contextuality and strong contextuality for measurement-based computers to accomplish the task of computing finite functions; note that these results are not concerned with universal quantum computation.  

Subsequently, we review how universal quantum computation is achieved fault-tolerantly in the circuit model via magic state distillation.  In this setting, contextual resources are not themselves directly measured, in contrast to the measurement-based model.  Rather, elementary operations are performed on a system composed of both an input quantum state and the resource state with the net effect that a non-elementary gate is performed on the input.  Contextuality is a necessary criterion for those states serving as useful resources.  By presenting paradoxes realised by computationally optimal magic states, we establish the relevance of the logical perspective to the task of universal quantum computation.

\begin{figure}[h]
  \begin{center}
  
\begin{tikzpicture}

\draw[fill=gray] (-3.75,2.1) rectangle (-3.25,2.5);
\draw  (-5,4.5) rectangle (-2,4);
\node at (-3.5,4.25) {CONTROL};

\node at (-3.575,3.25) {...};
\node at (-5.05,3.5) {\small $M_1$};
\node at (-3.05,3.5) {\small $M_n$};
\node at (-4,3) {\small $o_1$};
\node at (-2,3) {\small $o_n$};

\node (v2) at (-4.5,2.3) {\small $M_1 \to o_1$};
\node at (-2.5,2.3) {\small $M_n \to o_n$};

\draw  (-5.25,2.5) rectangle (-3.75,2.1);
\draw  (-3.25,2.5) rectangle (-1.75,2.1);
\draw[-triangle 60] (-4.75,4) -- (-4.75,2.5);
\draw [-triangle 60] (-4.25,2.5) -- (-4.25,4) ;
\draw [-triangle 60] (-2.75,4) -- (-2.75,2.5);
\draw [-triangle 60] (-2.25,2.5) -- (-2.25,4) ;

\draw [-triangle 60](-4.5,5) -- (-4.5,4.5);
\draw [-triangle 60](-2.5,4.5) -- (-2.5,5);
\node at (-4.5,5.25) {INPUT};
\node at (-2.5,5.25) {OUTPUT};
\end{tikzpicture}  
  
  \end{center} 
  \caption{Computation of functions by sequential measurements on a resource system.  The input and previous outcomes are processed to choose the next measurement.  Having performed all the measurements, the input and the outcomes are processed to give the output.  All processing is done by a restricted control computer capable of performing mod-2 arithmetic \cite{raussendorf}}.
  \label{MBQC}
\end{figure}
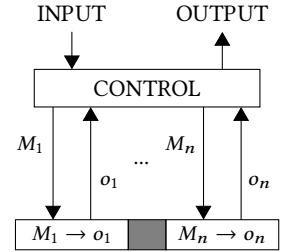

\subsection{Measurement-based computation}

Anders and Browne \cite{anders} proved that a computer capable only of mod-2 addition supplemented with access to measurements on a GHZ state can deterministically compute the OR of two bits (a nonlinear function).  The input bits determine the measurements performed on the three qubits; linear post-processing of the measurement outcomes yields the desired output.  Raussendorf \cite{raussendorf} generalised this to mod-2 linear measurement-based computers.  A computer restricted to performing mod-2 addition that can control measurements on an empirical model can compute a nonlinear function with high probability only if the empirical model is contextual.  The computation's failure probability is lower bounded by the product of the resource state's noncontextual fraction \cite{abramsky2017contextual} and a measure of the function's nonlinearity.  So, \emph{deterministically} computing a nonlinear function requires a \emph{strongly} contextual resource.

\subsection{The magic state model}

Fault-tolerant quantum computing works by supplementing stabilizer operations (that are easy to both implement and classically simulate) with access to magic states.  We briefly summarize elements of this theory here.  The allowed measurements are Weyl operators built from spin operators.  Following Gross' phase space formalism \cite{gross2006hudson}, the dimension $d$ is an odd prime number.  The single-qudit Pauli spin matrices are defined by $$X(q) \ket{j} = \ket{j + q} \quad \text{ and } \quad Z(p) \ket{j} = \om^{p j} \ket{j}$$ where $$\om = e^{2 \pi i / d}$$ and addition is modulo $d$.  A Weyl operator, represented by \emph{phase point }coordinates in $\Z_d^{2}$, is defined by $$W(p, q) = \om^{-2^{-1}pq} Z(p) X(q)$$ where $2^{-1}$ is the multiplicative inverse of $2$ in $\Z_d$.  An $n$-particle Weyl operator, represented by coordinates in $\Z_d^{2n}$, is defined as $$W(p_1, q_1, ..., p_n, q_n) = W(p_1, q_1) \otimes ...  \otimes  W(p_n, q_n).$$
For convenience, we denote a point $(p_1, q_1, ..., p_n, q_n)$ in an $n$-particle phase space as $(\textbf{p}, \textbf{q})$.  The \emph{symplectic product} of two phase space points is defined by: $$[(\textbf{p}, \textbf{q}),(\textbf{p'}, \textbf{q'})] = \sum_{i=1}^n p_i q'_i - p'_i q_i.$$  Weyl operators obey a composition law: $$W(\textbf{p},\textbf{q}) W(\textbf{p'},\textbf{q'}) = \om^{2^{-1} [(\textbf{p},\textbf{q}), (\textbf{p'},\textbf{q'})]} W(\textbf{p}+\textbf{p'}, \textbf{q}+\textbf{q'}).$$  Therefore, $$W(\textbf{p},\textbf{q}) W(\textbf{p'},\textbf{q'}) = W(\textbf{p'},\textbf{q'}) W(\textbf{p},\textbf{q}) = W(\textbf{p}+\textbf{p'}, \textbf{q}+\textbf{q'})$$ if and only if $$[(\textbf{p},\textbf{q}), (\textbf{p'},\textbf{q'})] = 0.$$  

The $n$-Weyl operators (with phases) form the $n$-Pauli group: $$\cc_1^n = \{\om^{2^{-1}m} W(\textbf{p},\textbf{q}) : m \in \Z_{2d}, (\textbf{p},\textbf{q}) \in \Z_d^{2n}\}.$$  The Clifford gates are those unitaries preserving the Pauli group: $$\cc_2^n = \{U : U \cc_1^n U^\dagger \subseteq \cc_1^n\}.$$  The Clifford hierarchy \cite{gotthesis} is defined inductively:  $$\cc_k^n = \{U : U \cc_1^n U^\dagger \subseteq \cc_{k-1}^n\}.$$  
It forms a nested sequence under inclusion:  $$\cc_1^n \subset \cc_2^n \subset \cc_3^n \subset ... .$$
The magic states we will consider arise from the strict third level of the Clifford hierarchy: $\cc^n_3 \setminus \cc^n_2$.  Cui, Gottesman, and Krishna \cite{cui2017diagonal}, building on work of Campbell \cite{campbell2014enhanced} and Howard-Vala \cite{howard2012qudit}, give an explicit description of all diagonal gates in the $k^{\text{th}}$ level.  (We note that not much is presently known about the nondiagonal gates of the Clifford hierarchy.)  For $d > 3$, diagonal gates of the third level are those of the form: $$ U_\Phi = \sum_{\textbf{j} \in \Z_d^n} \om^{\Phi(\textbf{j})} \ket{\textbf{j}} \bra{\textbf{j}}$$ where $\Phi$ is an $n$-variable polynomial over $\Z^d$ of degree $3$.  Every strict third-level gate yields a magic state $$\ket{\Phi} = U_\Phi \ket{+}^{\otimes n} = d^{-n} \sum_{\textbf{j} \in \Z_d^n} \om^{\Phi(\textbf{j})} \ket{\textbf{j}}.$$  As an example, the two-qutrit controlled-$S$ gate (where $S = T^2$ and $T$ is a higher-dimensional generalisation \cite{howard2012qudit} of the $\pi/8$ gate) is $U_{j_1 j_2^2} \in \cc^2_3$.  Its corresponding magic state was computationally found to be strongly contextual in \cite{de2017graph}. 

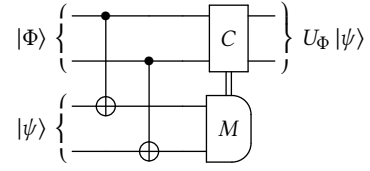
\begin{figure}[h]
\begin{center}

\hspace{2em}
\Qcircuit @C=1em @R=1em {
\lstick{} & \ctrl{2} & \qw & \qw & \multigate{1}{C} & \qw && \dstick{\quad U_\Phi\ket{\psi}}  \\
\lstick{} & \qw & \ctrl{2} & \qw & \ghost{C} & \qw & \rstick{} 
\inputgroupv{1}{2}{.9em}{.96em}{\ket{\Phi}}\\
\lstick{} & \targ & \qw & \qw & \multimeasureD{1}{M} \cwx[-1]  \\
\lstick{} & \qw & \targ & \qw & \ghost{M}
\inputgroupv{3}{4}{.9em}{.98em}{\ket{\psi}} \\
{\gategroup{1}{6}{2}{6}{1em}{\}}}
}
\end{center}
\caption{Schematic of a stabilizer circuit deterministically implementing a diagonal gate in the third level of the Clifford hierarchy \cite{howard2017application}.  CNOT gates entangle the input state and the magic state before a stabilizer measurement is made on the input.  A Clifford correction, conditioned on the measurement outcome, is performed on the magic state lines resulting in the input having a gate applied to it.}
\end{figure}

Magic states arising from diagonal third-level gates $U_\Phi \in \cc^n_3 \setminus \cc^n_2$ are optimal for enabling quantum universality in that they admit direct deterministic protocols for injecting the gate $U_\Phi$ \cite{howard2017application}.  Higher-level gates require more complex injection protocols involving access to EPR pairs and ancillary resources.  Gates from outside the Clifford hierarchy may be implemented via state injection; however, this is a stochastic process that requires randomly many attempts.

\section{Paradoxical magic states}
\subsection{Summary of results}

We now define strong magic states as those corresponding to two-variable cubic polynomials with \emph{no local terms}, i.e. the coefficients for the $j^3, k^3$ terms vanish, and show that they exhibit paradoxes with respect to stabilizer measurements: $$\cm = \{W(\textbf{p},\textbf{q}) : (\textbf{p},\textbf{q}) \in \Z_d^{2n}\}.$$  The contexts are the maximal commutative subsets of $\cm$.  Thus, we are concerned with \emph{bona fide} contextuality rather than  nonlocality.  The possible outcomes $\co$ of measuring a Weyl operator is $\Z_d$ with $k \in \Z_d$ labelling the eigenvalue $\omega^k$. Proofs are found in the following subsection.

\begin{defn}
A two-qudit magic state $$\ket{\Phi} = d^{-2} \sum_{j,k \in \Z_d^2} \om^{\Phi(j,k)} \ket{j}\ket{k} $$ is \emph{strong} if $$\Phi(j,k) = \phi_1 j^2 k + \phi_2 j k^2 + q(j,k)$$ with $q(j,k)$ being any quadratic polynomial and either $\phi_1$ or $\phi_2 \not\equiv 0$.
\end{defn}

These states are capable of injecting a gate built from arbitrary combinations of controlled-S gates on either qudit.  The restrictions on local terms are not strictly necessary; they strike a balance between sufficient generality to establish a clear pattern whilst admitting tractable analytic arguments.  Indeed, computations in low dimensions are consistent with the possibility that \emph{all} magic states for diagonal $\cc^2_3 \setminus \cc^2_2$ gates in prime dimension 2 mod 3 are strongly contextual.

\begin{restatable}{thm}{main}
\label{thm:mainthm}
Suppose that the prime dimension $d \not\equiv 1 \;(\textrm{mod}\ 3)$.  
  All strong magic states $\ket{\Phi}$ are strongly contextual with respect to stabilizer measurements.  It follows immediately that the states $C \ket{\Phi}$, where $\ket{\Phi}$ is strong and $C$ is any Clifford gate, are also strongly contextual.
\end{restatable}

We prove that for any strong magic state $\ket{\Phi}$, any Clifford gate $C$, and any hidden variable $\lambda: \Z_d^4 \to \Z_d$, $\lambda$ is inconsistent with $C \ket{\Phi}$.  After reducing to the case where $\phi_1 \not\equiv 0$, $q(j,k) = 0$, and $C = \I$, we show that $\lambda$ must predict the occurrence of an impossible joint outcome for one of the following $d(d+1)$ pairs of commuting measurements.

\begin{table}[H]

\label{tab:contexts}
\begin{tabular}{lllll}
\textbf{Type\quad}                 & \textbf{Operators}                                                    & \textbf{Phase points}                                 &  &  \\
$\mathrm{I}_\alpha$           & $Z \otimes \I$ and $\I \otimes Z^\alpha X$                   & $(1,0,0,0)$ and $(0,0,\alpha,1)$             &  &  \\
$\mathrm{II}_\alpha$          & $\I \otimes Z$ and $Z^\alpha X \otimes \I$                   & $(0,0,1,0)$ and $(\alpha, 1,0,0)$            &  &  \\
$\mathrm{III}_{\alpha,\beta}$ & $Z \otimes Z^\beta$ and $X \otimes Z^\alpha X^{-\beta^{-1}}\;$ & $(1,0,\beta, 0)$ and $(0,1,\alpha, -\beta^{-1})$ &  & 
\end{tabular}
\caption{The three families of contexts needed for our argument.  Here, $\alpha, \beta \in \Z_d$ and $\beta \neq 0$.  Operator phases omitted. }
\end{table}

The following, surprising lemma allows us to dramatically reduce (from exponentially to polynomially many) the number of hidden variables $\lambda: \Z_d^{2n} \to \Z_d$ we must consider.  It was established earlier by Delfosse  \emph{et al.} \cite{delfosse2016}; we give an independently obtained, algebraic proof in the following subsection.
\begin{restatable}{lemma}{lineargs}
\label{lem:lineargs}
Suppose that $n \geq 2$ and that $\lambda : \Z_d^{2n} \to \Z_d$ is a hidden variable for $n$-particle stabilizer quantum mechanics that is consistent with a quantum state.  Then, $\lambda(\textbf{p},\textbf{q}) = \bm{ \lambda } \cdot (\textbf{p},\textbf{q})$ for some $\bm{\lambda} \in \Z_d^{2n}$.
\end{restatable}

We then establish a number-theoretic criterion for the impossibility of a joint outcomes upon performing measurements on a state $\ket{\Phi}$.  A \emph{permutation polynomial} is a (possibly multivariable) polynomial over $\Z_d$ that takes each value in $\Z_d$ equally many times:  $$\forall \textbf{x} \in \Z_d^n, \quad |p^{-1}(\textbf{x})| = d^{n-1}.$$  In general, the problem of determining whether a polynomial over a finite field is a permutation polynomial is a difficult one; its study remains an active area of number theory \cite{hou2015permutation}.

\begin{restatable}{lemma}{permpoly}
\label{lem:permpoly}
The joint outcome $(A,B)$ is impossible for the measurement of $U = W(\textbf{p}, \textbf{q})$ and $V = W(\textbf{p'}, \textbf{q'})$ with $[U,V] = 0$ on the state $\ket{\Phi}$ if and only if, for all $(j,k) \in \Z^2_d$, the polynomial $\Psi(x,y)$ is a permutation polynomial where
\begin{equation*}
\begin{split}
\Psi(x,y) &= - x A - y B - 2^{-1}((x p_1 + y p'_1) (x q_1 + y q'_1) + (x p_2  \\ 
&\quad\; + y p'_2) + (x q_2 + y q'_2)) +  j (x p_1 + y p'_1) + k (x p_2 + y p'_2) \\ 
&\quad\; + \Phi(j - (x q_1 + y q'_1),k - (x q_2 + y q'_2)).
\end{split}
\end{equation*}
\end{restatable}

We will require elements of Dickson's classification of one-variable permutation polynomials over finite fields of low degree \cite{dickson1896analytic}.

\begin{thm}[Dickson, 1896]
\label{thm:dickson}
Suppose $d$ is a prime power, $d \not\equiv 1 \;(\textrm{mod}\ 3)$ and $f: \Z_d \to \Z_d$ is a polynomial of degree at most 3.  Then, $f$ is a permutation polynomial if and only if $f(x) = ag(x + b) + c$ where $a \neq 0$ and $g(x) = x$ or $x^3$.
\end{thm}

By Lemma \ref{lem:lineargs}, a hidden variable $\lambda$ prescribes to $W(p_1, q_1, p_2, q_2)$ the outcome $\lambda_1 p_1 + \lambda_2 q_1 + \lambda_3 p_2 + \lambda_4 q_2$.  By substituting the phase points of our chosen measurements, the coefficients of the polynomial describing the state $\ket{\Phi}$, and the outcomes prescribed by a hidden variable $\lambda$ into the master equation  of Lemma \ref{lem:permpoly}, we see that inconsistency of $\lambda$ with $\ket{\Phi}$ for the measurements of Table 1 is equivalent to one of the following $d(d+1)$ polynomials being a permutation polynomial for all possible choices of $j,k \in \Z_d$:
\begin{align*}
\Psi_{\mathrm{I}_\alpha}(x,y) &= x (j - \lambda_2) + y^2 (j \phi_2 - 2^{-1}\alpha) + y(\alpha (k - \lambda_4) - \lambda_3  \\ 
&\quad\; - j^2 \phi_1 - 2 j k \phi_2) \\
\\
\Psi_{\mathrm{II}_\alpha}(x,y) &= x (k - \lambda_4) + y^2 (k \phi_1 - 2^{-1}\alpha) + y(\alpha (j - \lambda_2) - \lambda_1  \\ 
&\quad\; - k^2 \phi_2 - 2 j k \phi_1) \\
\\
\Psi_{\mathrm{III}_{\alpha,\beta}}(x,y) &= x (j - \lambda_2 + \beta(k - \lambda_4) ) + y^3 (\beta^{-1}(\phi_1 - \beta^{-1}\phi_2)  \\ 
&\quad\; + y^2 (\beta^{-1}(2^{-1}\alpha - 2 j \phi_1 - 2 k \phi_2 + \beta^{-1} j \phi_2) + k \phi_1 )  \\ 
&\quad\; + y(\alpha(k - \lambda_4)+ \beta^{-1}(\lambda_3 + j^2 \phi_1 + 2 j k \phi_2) - \lambda_1  \\ 
&\quad\; - 2 j k \phi_1 - k^2 \phi_2 )  
\end{align*}

In the next subsection, we prove by inspection that this must hold and, thus, that no hidden variable is consistent with $\ket{\Phi}$.  We conclude that all strong magic states are strongly contextual with respect to stabilizer operations.

\subsection{Proofs of main results}

Here, we provide explicit details of the proofs of the main results.  First, we establish that the algebraic relations between commuting Weyl operators enforce a strong condition on hidden variables: they must be group homomorphisms from phase space to outcomes.  This dramatically reduces (from $d^{d^{2n}}$ functions to $d^{2n}$ homomorphisms) the number of hidden variables we need to consider.  The assumption of multiple qudits is crucial here.

\lineargs*

\begin{proof}
The hidden variable $\lambda$ prescribes the outcome $\lambda(\textbf{p},\textbf{q})$ to $W(\textbf{p},\textbf{q})$.  If $\lambda$ is consistent with a quantum state, these outcomes must respect the algebraic relations between commuting Weyl operators: $$\lambda(\textbf{p},\textbf{q}) + \lambda(\textbf{p'},\textbf{q'}) = \lambda(\textbf{p+p'},\textbf{q+q'})$$ whenever $$[(\textbf{p},\textbf{q}),(\textbf{p'},\textbf{q'})] = 0.$$  Thus, when $n = 2$:  $$\lambda(p_1, q_1, p_2, q_2) = \lambda(p_1, q_1, 0,0) + \lambda(0,0, p_2, q_2)$$ and, whenever $p_1 q_1 = -p_2 q_2$, we have that \begin{align*}\lambda(p_1, q_1, p_2, q_2) &= \lambda(p_1, 0, p_2, 0) + \lambda(0, q_1, 0, q_2) \\
&= \lambda(p_1, 0, 0, 0) + \lambda(0, 0, p_2, 0) \\ & \quad + \lambda(0, q_1, 0, 0) + \lambda(0, 0, 0, q_2).\end{align*}  Thus, $\lambda$ is linear in phase coordinates for the first qudit:
\begin{align*}
& \lambda(1, k, 0, 0) \\
&= \lambda(1, k, 0, -2^{-1}k) + \lambda(0,0,0,2^{-1}k) \\
 &= [\lambda(1, 2^{-1}k, 1, -2^{-1}k) + \lambda(0,2^{-1}k,-1,0)] + \lambda(0,0,0,2^{-1}k) \\
 &= [\lambda(1, 0,0,0) + \lambda(0, 2^{-1}k, 0,0) + \lambda(0, 0, 1, 0) + \lambda(0,0,0, -2^{-1}k)] \\
 & \quad\quad\quad\quad\quad\quad\; + [\lambda(0,2^{-1}k,0,0) + \lambda(0,0,-1,0)] + \lambda(0,0,0,2^{-1}k) \\ 
  &= \lambda(1, 0,0, 0) + \lambda(0, k,0, 0)
\end{align*}
Similarly, $\lambda(0,0,1, k) = \lambda(0,0,1,0) + \lambda(0,0,0,k)$.  Thus, $\lambda$ is linear.  For $n > 2$, the same argument holds.
\end{proof}

Next, we establish a master equation governing the possibility vs. impossibility of observing a given joint outcome upon measurement of a pair of commuting Weyl measurements on a state $\ket{\Phi}$.  The lemma is easily extended to more qudits/measurements and to polynomials $\Phi$ of any degree and holds for any prime dimension $d$.

\permpoly*

\begin{proof}
The eigenvalues of a Weyl operator $W$ are $\om^k$ for $k \in \Z_d$.  The projection onto the 1-eigenspace of $W$ is given by $$\Pi(0|W) = d^{-1}\sum_{x \in \Z_d} W^x.$$  This can be seen by noting that $W^d = \mathbb{I}$, $W^\dagger = W^{-1}$.  Therefore, the projection onto the $\om^k$-eigenspace of $W$ is given by $$\Pi(k|W) = d^{-1}\sum_{x \in \Z_d} \om^{-xk}W^x.$$

Measuring $U = W(\textbf{p}, \textbf{q})$ and $V = W(\textbf{p'}, \textbf{q'})$ with $[U,V] = 0$ and obtaining the joint outcome $(A,B) \in \Z_d^2$ corresponds to the projection: 

\begin{align*}
\Pi(A,B|U,V) &=  \Pi(A|U) \Pi(B|V)  \\
&= d^{-2} (\sum_{x \in \Z_d} \om^{- x A} U^x)(\sum_{y \in \Z_d} \om^{-y  B} V^y) \\
&= d^{-2} \sum_{x,y \in \Z_d} \om^{- x A - y B} U^x V^y \\
&= d^{-2} \sum_{x,y \in \Z_d} \om^{- x A - y B} W(\textbf{P},\textbf{Q}) \\
&= d^{-2} \sum_{x,y \in \Z_d} \om^{- x A - y B-2^{-1}(\textbf{P} \cdot \textbf{Q})} Z(P_1) X(Q_1) \otimes Z(P_2) X(Q_2)
\end{align*}
where $(\textbf{P},\textbf{Q})$ is shorthand for $x(\textbf{p},\textbf{q}) + y(\textbf{p'},\textbf{q'})$.

Applying this to a two-qudit magic state $\ket{\Phi} = d^{-2} \sum_{j,k \in \Z_d} \om^{\Phi(j,k)} \ket{j}\ket{k}$, we obtain:

\begin{strip}
\par\noindent\rule[8pt]{0.5 \textwidth}{0.4pt}
\begin{align*}
\hspace{7em} \Pi(A,B|U,V)\ket{\Phi} &= d^{-4} \sum_{x,y \in \Z_d} \sum_{j,k \in \Z_d} \om^{- x A - y B - 2^{-1}(P_1 Q_1 + P_2 Q_2)+\Phi(j,k)} Z(P_1) X(Q_1) \otimes Z(P_2) X(Q_2) \ket{j}\ket{k} \\
&= d^{-4} \sum_{x,y \in \Z_d} \sum_{j,k \in \Z_d} \om^{- x A - y B - 2^{-1}(P_1 Q_1 + P_2 Q_2)+\Phi(j,k)} Z(P_1) \otimes Z(P_2) \ket{j + Q_1}\ket{k + Q_2} \\
&= d^{-4} \sum_{x,y \in \Z_d} \sum_{j,k \in \Z_d} \om^{- x A - y B - 2^{-1}(P_1 Q_1 + P_2 Q_2)+\Phi(j - Q_1,k - Q_2)} Z(P_1) \otimes Z(P_2) \ket{j}\ket{k} \\
&= d^{-4} \sum_{j,k \in \Z_d} \sum_{x,y \in \Z_d} \om^{- x A - y B - 2^{-1}(P_1 Q_1 + P_2 Q_2) + j P_1 + k P_2 +\Phi(j - Q_1,k - Q_2)} \ket{j}\ket{k} 
\end{align*}
\par\noindent\hspace{0.5 \textwidth}\rule{0.5 \textwidth}{0.4pt}
\end{strip}

Observing the joint outcome $(A,B)$ for the measurements $U,V$ on the state $\ket{\Phi}$ is impossible precisely when $\Pi(A,B|U,V)\ket{\Phi} = 0$; that is, when the terms $$\sum_{x,y \in \Z_d} \om^{- x A - y B - 2^{-1}(P_1 Q_1 + P_2 Q_2) + j P_1 + k P_2 +\Phi(j - Q_1,k - Q_2)} $$ vanish for all $j,k \in \Z_d$.  Such a term is a sum of $d^2$ many primitive $d^\text{th}$-roots of unity and, since $d$ is prime, it vanishes if and only if each $d^\text{th}$-root appears $d$ many times.  Thus, impossibility of the measurement outcome is equivalent to 
\begin{equation*}
\begin{split}
\Psi(x,y) &= - x A - y B - 2^{-1}((x p_1 + y p'_1) (x q_1 + y q'_1) + (x p_2  \\ 
&\quad\; + y p'_2) + (x q_2 + y q'_2)) +  j (x p_1 + y p'_1) + k (x p_2 + y p'_2) \\ 
&\quad\; + \Phi(j - (x q_1 + y q'_1),k - (x q_2 + y q'_2))
\end{split}
\end{equation*}
being a permutation polynomial in $x,y$ for all $j,k \in \Z_d$.
\end{proof}

We are now ready to prove our main theorem.

\main*

\begin{proof}

Suppose \begin{align*}
\ket{\Phi} &= d^{-2} \sum_{j,k \in \Z_d^2} \om^{\Phi(j,k)} \ket{j}\ket{k}  \\
&= d^{-2} \sum_{j,k \in \Z_d^2} \om^{\phi_1 j^2 k + \phi_2 j k^2 + q(j,k)} \ket{j}\ket{k}
\end{align*}
is a strong magic state and $C$ is any Clifford gate.  As Clifford gates act as permutations on the Pauli operators under conjugation, they preserve all properties of contextuality with respect to stabilizer measurements.  Therefore, by constructing the diagonal Clifford gate $U_q = \sum_{j,k \in \Z_d^2} \om^{q(j,k)} \ket{j}\ket{k} \bra{k}\bra{j}$, we can apply the Clifford gate $(CU_q)^\dagger$ to $C \ket{\Phi}$ and find that $C \ket{\Phi}$ is strongly contextual if and only if the state $d^{-2} \sum_{j,k \in \Z_d^2} \om^{\phi_1 j^2 k + \phi_2 j k^2} \ket{j}\ket{k}$ is.  So, without loss of generality, $C = \I$ and $\Phi(j,k) = \phi_1 j^2 k + \phi_2 j k^2$ with either $\Phi_1 \not\equiv 0$ or $\Phi_2 \not\equiv 0$.  By possibly swapping qudits, we can assume that $\Phi_1 \not\equiv 0$.

To prove that $\Ket{\Phi}$ is strongly contextual, we must prove that every possible hidden variable $\lambda: \Z_d^{2n} \to \Z_d$ is inconsistent with  $\Ket{\Phi}$.  Given a pair of state and hidden variable, we will show that for one of the contexts of Table 1, the joint outcome predicted by $\lambda$ to occur will, in fact, have probability zero of being the result of performing the measurements of that context on the state $\ket{\Phi}$.  More specifically, we will show that if a $\lambda$ predicts joint outcomes with nonzero probability for all contexts of type $\text{I}_\alpha$ and $\text{II}_\alpha$, it must predict a joint outcome with probability zero for the measurements of some context of type $\text{III}_{\alpha,\beta}$.

 \setcounter{table}{0}
\begin{table}[H]
\centering
\begin{tabular}{lllll}
\textbf{Type\quad}                 & \textbf{Operators}                                                    & \textbf{Phase points}                                 &  &  \\
$\mathrm{I}_\alpha$           & $Z \otimes \I$ and $\I \otimes Z^\alpha X$                   & $(1,0,0,0)$ and $(0,0,\alpha,1)$             &  &  \\
$\mathrm{II}_\alpha$          & $\I \otimes Z$ and $Z^\alpha X \otimes \I$                   & $(0,0,1,0)$ and $(\alpha, 1,0,0)$            &  &  \\
$\mathrm{III}_{\alpha,\beta}$ & $Z \otimes Z^\beta$ and $X \otimes Z^\alpha X^{-\beta^{-1}}\;$ & $(1,0,\beta, 0)$ and $(0,1,\alpha, -\beta^{-1})$ &  & 
\end{tabular}
\caption{The three families of contexts needed for our argument.  Here, $\alpha, \beta \in \Z_d$ and $\beta \neq 0$.  Operator phases omitted.}
\end{table}

By Lemma \ref{lem:lineargs}, $\lambda$ prescribes to $W(p_1, q_1, p_2, q_2)$ the outcome $\lambda_1 p_1 + \lambda_2 q_1 + \lambda_3 p_2 + \lambda_4 q_2$.  By substituting the phase points of our chosen measurements, the coefficients of the polynomial describing the state $\ket{\Phi}$, and the outcomes prescribed by a hidden variable $\lambda$ into the master equation  of Lemma \ref{lem:permpoly}, we see that inconsistency of $\lambda$ with $\ket{\Phi}$ for the measurements of Table 1 is equivalent to one of the following $d(d+1)$ polynomials being a permutation polynomial for all possible choices of $j,k \in \Z_d$:
\begin{align*}
\Psi_{\mathrm{I}_\alpha}(x,y) &= x (j - \lambda_2) + y^2 (j \phi_2 - 2^{-1}\alpha) + y(\alpha (k - \lambda_4) - \lambda_3  \\ 
&\quad\; - j^2 \phi_1 - 2 j k \phi_2) \\
\\
\Psi_{\mathrm{II}_\alpha}(x,y) &= x (k - \lambda_4) + y^2 (k \phi_1 - 2^{-1}\alpha) + y(\alpha (j - \lambda_2) - \lambda_1  \\ 
&\quad\; - k^2 \phi_2 - 2 j k \phi_1) \\
\\
\Psi_{\mathrm{III}_{\alpha,\beta}}(x,y) &= x (j - \lambda_2 + \beta(k - \lambda_4) ) + y^3 (\beta^{-1}(\phi_1 - \beta^{-1}\phi_2)  \\ 
&\quad\; + y^2 (\beta^{-1}(2^{-1}\alpha - 2 j \phi_1 - 2 k \phi_2 + \beta^{-1} j \phi_2) + k \phi_1 )  \\ 
&\quad\; + y(\alpha(k - \lambda_4)+ \beta^{-1}(\lambda_3 + j^2 \phi_1 + 2 j k \phi_2) - \lambda_1  \\ 
&\quad\; - 2 j k \phi_1 - k^2 \phi_2 )  
\end{align*}

We have dropped all terms constant in $x,y$ as the translation of a permutation polynomial is still a permutation polynomial.  Notice that all the above polynomials $p(x,y)$ in $x,y$ are the sum of two single-variable polynomials $p_x(y),p_y(y)$ in $x$ and $y$: $p(x,y) = p_x(x) + p_y(y)$.  If $p_x(x)$ is a single-variable permutation polynomial, it follows that $p(x,y)$ is a two-variable permutation polynomial.  Our proof thereby proceeds by closely analysing these polynomials and repeatedly applying Dickson's theorem (Theorem 2).

First, consider $\Psi_{\mathrm{I}_\alpha}(x,y)$ for  $\alpha = 2 \lambda_2 \phi_2$.  If $\lambda$ is consistent with $\ket{\Phi}$, there must exist a pair $j,k \in \Z_d$ for which $\Psi_{\mathrm{I}_\alpha}$ is \emph{not} a permutation polynomial in $x,y$.  When $j \neq \lambda_2$, $\Psi_{\mathrm{I}_\alpha}(x,y)$ is the sum of a nonzero linear term in $x$ (a single-variable permutation polynomial) and a polynomial in $y$; thus, it is a permutation polynomial in $x,y$.  Therefore, $j = \lambda_2$.  Making this substitution, the $x$ term vanishes and, by our choice of $\alpha$, so does the coefficient for the $y^2$ term.  As $\Psi_{\mathrm{I}_\alpha}$ is simply linear in $y$ and, by assumption, not a permutation polynomial,  its linear coefficient must vanish.  Thus, either $\lambda$ is inconsistent with $\ket{\Phi}$ or $\lambda_3 = -\lambda_2(2 \lambda_4 \phi_2 + \lambda_2 \phi_1)$.

By a similar analysis of $\Psi_{\mathrm{II}_\alpha}(x,y)$, we find that by choosing $\alpha = 2 \lambda_4 \phi_1$, either $\lambda$ is inconsistent with $\ket{\Phi}$ or $\lambda_1 = -\lambda_4(2 \lambda_2 \phi_1 + \lambda_4 \phi_2)$.

Finally, we make the substitutions for $\lambda_1,\lambda_3$ in $\Psi_{\mathrm{III}_{\alpha,\beta}}(x,y)$ and note that, as they too are linear in $x$, they are permutation polynomials unless $j = \lambda_2 - \beta( k - \lambda_4)$.  For convenience, we multiply the resulting expression by $\beta \neq 0$: \begin{align*}
&y^3 (\phi_1 - \beta^{-1} \phi_2) + y^2(2^{-1} \alpha + 3 k (\beta \phi_1 - \phi_2) + \beta^{-1} \lambda_2 \phi_2 + \lambda_4 \phi_2 \\ &- 2 \beta \lambda_4 \phi_1 - 2 \lambda_2 \phi_1 ) + y( 3 \beta k^2(\beta \phi_1 - \phi_2) + k (\beta  (\alpha -4 \lambda_2 \phi _1 \\ &+ 2 \lambda_4 \phi _2) - 4 \beta^2 \lambda_4 \phi_1 + 2 \lambda_2 \phi _2) + \beta (-\alpha \lambda_4 + \lambda_4^2 \phi_2 + 4 \lambda_2 \lambda_4 \phi_1) \\ &+\beta^2 \lambda_4^2 \phi_1 - 2 \lambda_2 \lambda_4 \phi_2 )
\end{align*}

We consider two cases.  First, if $\phi_2 \equiv -1$, then, by choosing $\alpha = 6(\lambda_2 \phi_1 - \lambda_4)$ and $\beta = \phi_1^{-1}$, the resulting polynomial factors as:$$ 2 (y + \phi_1^{-1}(k - \lambda_4))^3.$$

However, if $\phi_2 \not\equiv -1$, we may choose $\alpha = 2(\phi_2 + 1)^{-1}(\lambda_2 \phi_1 (\phi_2 + 2) + \lambda_4 (\phi_2^2 - 1))$ and $\beta = \phi_1^{-1}(\phi_2 + 1)$; the resulting polynomial factors as:$$(y + \phi_1^{-1}(k - \lambda_4)(\phi_2+1))^3.$$  In both cases, the type $\mathrm{III}$ polynomial, for our choices of $\alpha$ and $\beta$, is a permutation polynomial for all $j,k \in \Z_d$ and, therefore, $\lambda$ is inconsistent with $\ket{\Phi}$. We conclude that $\ket{\Phi}$ is strongly contextual with respect to two-qudit stabilizer measurements.\end{proof}

\section{Conclusions}

The surprising finding of an abundance of exotic quantum paradoxes (of arbitrarily large dimension) among those magic states that are optimal as resources for quantum computation bolsters a refinement of the resource theory of contextuality that emphasises the computational power of logical paradoxes.  This intuitive hypothesis for the origin of the computational advantages provided by contextuality emerges from observing the pattern across informatic settings of determinstic advantages enabled by strongly contextual resources together with the view of standard contextuality as a stochastic mixture of classical and paradoxical parts \citep{abramsky2017contextual}.  It is rendered more compelling by its applicability in the experimentally tractable setting of universal quantum computation via state injection.

A key desideratum of a property of systems considered to be a resource is that the benefits conferred scale with the strength of that property.  Thus, our results validate both the idea of contextuality as an essential resource in quantum computation and seeing logical paradoxes as a highly useful resource-theoretic notion of extremal contextuality among many potential measures of contextuality.

Outside of purely communicational tasks, it is in the setting of measurement-based computation where the centrality of logical paradoxes to enabling advantage has been made most clear.  Here, the classical part  of a resource offers no benefits; it is the paradox that does all the work.  Our results differ from this setting in two critical ways.  There, paradoxes enable enlarging the class of finite functions that can be computed (with resources that may be quantum) whereas our interest is in enabling full-blown universal quantum computation in the circuit model (with resources that must be quantum).  

The second key difference is that measurement-based computation, framed entirely in terms of measurements on a resource state, is naturally suited to analysis in terms of contextuality.  In the circuit-based model, it is \emph{transformations} of a system that performs the computational work which adds a layer of conceptual obstacles to a fine-grained analysis of how contextuality powers computation.   Nevertheless, the salience of the logical perspective survives in the complex setting of quantum computation.

\subsection{Future work}
We indicate questions for future work in two main directions: along purely logical lines towards developing a structural understanding of how contextuality enables computational advantage, and along concretely quantum computational lines.

The most immediate task is to precisely delineate the mechanism by which paradoxes aid in achieving computational advantage. This is most clearly done in a purely logical setting, without the complications and excessive assumptions of choosing a physical theory.  By interpreting these results within quantum theory, such logical insights can be practically applied. As this work takes place in the circuit model of quantum computing which explicitly involves transformations of a physical system (in contrast to the measurement-based setting of previous work) this raises the question of how to incorporate dynamical aspects into frameworks of contextual logic.  Given the adaptivity of measurements in measurement-based computational models, there is the need to incorporate causality into frameworks of contextual logic.  One might then ask:  what is the interaction between causal structures on measurements and the aforementioned dynamical aspects?  Furthermore, can the distillation of weak contextuality into strong computational resources, exemplified in the quantum setting by the magic state distillation protocol, be understood at a logical level?

At a concrete level, our work leaves open the question of a complete classification of quantum paradoxes in qudit stabilizer mechanics.  It also raises questions about the exact nature of the mutual relationships between the Clifford hierarchy,  quantum resources for deterministic injection, and strong contextuality.  For example, it is reasonable to conjecture that every strongly contextual two-qudit state admits a deterministic injection protocol enabling universal quantum computation.  It remains open whether these magic state-based paradoxes admit cohomological witnesses \cite{ccp, okay2017topological} or whether it is possible to relate the contextual fraction of a magic state to the difficulty or length of its associated protocols for distillation/injection.

\begin{acks}                            
We thank Samson Abramsky, Rui Soares Barbosa, and Eric Cavalcanti for discussions on contextuality and Mark Howard for discussions on the Clifford hierarchy and state injection.  Financial support was provided by the EPSRC via grant EP/N017935/1 (Contextuality as a resource in quantum computation).  This work was partially conducted during the Logical Structures in Computation programme at the Simons Institute for the Theory of Computing at the University of California, Berkeley.
\end{acks}

\bibliography{licscontextuality}


\begin{thebibliography}{59}


\ifx \showCODEN    \undefined \def \showCODEN     #1{\unskip}     \fi
\ifx \showDOI      \undefined \def \showDOI       #1{#1}\fi
\ifx \showISBNx    \undefined \def \showISBNx     #1{\unskip}     \fi
\ifx \showISBNxiii \undefined \def \showISBNxiii  #1{\unskip}     \fi
\ifx \showISSN     \undefined \def \showISSN      #1{\unskip}     \fi
\ifx \showLCCN     \undefined \def \showLCCN      #1{\unskip}     \fi
\ifx \shownote     \undefined \def \shownote      #1{#1}          \fi
\ifx \showarticletitle \undefined \def \showarticletitle #1{#1}   \fi
\ifx \showURL      \undefined \def \showURL       {\relax}        \fi
\providecommand\bibfield[2]{#2}
\providecommand\bibinfo[2]{#2}
\providecommand\natexlab[1]{#1}
\providecommand\showeprint[2][]{arXiv:#2}

\bibitem[\protect\citeauthoryear{Abramsky, Barbosa, Car\'{u}, de~Silva,
  Kishida, and Mansfield}{Abramsky et~al\mbox{.}}{2017b}]%
        {tqc}
\bibfield{author}{\bibinfo{person}{S. Abramsky}, \bibinfo{person}{R.~S.
  Barbosa}, \bibinfo{person}{G. Car\'{u}}, \bibinfo{person}{N. de Silva},
  \bibinfo{person}{K. Kishida}, {and} \bibinfo{person}{S. Mansfield}.}
  \bibinfo{year}{2017}\natexlab{b}.
\newblock \showarticletitle{Minimum quantum resources for strong non-locality}.
  In \bibinfo{booktitle}{{\em Theory Quant. Comp. Comm. Crypt. (TQC) 2017}}.
\newblock


\bibitem[\protect\citeauthoryear{Abramsky, Barbosa, de~Silva, and
  Zapata}{Abramsky et~al\mbox{.}}{2017c}]%
        {quantummonad}
\bibfield{author}{\bibinfo{person}{S. Abramsky}, \bibinfo{person}{R.~S.
  Barbosa}, \bibinfo{person}{N. de Silva}, {and} \bibinfo{person}{O. Zapata}.}
  \bibinfo{year}{2017}\natexlab{c}.
\newblock \showarticletitle{The quantum monad on relational structures}. In
  \bibinfo{booktitle}{{\em Math. Found. Comp. Sci. (MFCS) 2017}},
  Vol.~\bibinfo{volume}{83}. \bibinfo{pages}{35:1--35:19}.
\newblock


\bibitem[\protect\citeauthoryear{Abramsky, Barbosa, Kishida, Lal, and
  Mansfield}{Abramsky et~al\mbox{.}}{2016}]%
        {posspoly}
\bibfield{author}{\bibinfo{person}{S. Abramsky}, \bibinfo{person}{R.~S.
  Barbosa}, \bibinfo{person}{K. Kishida}, \bibinfo{person}{R. Lal}, {and}
  \bibinfo{person}{S. Mansfield}.} \bibinfo{year}{2016}\natexlab{}.
\newblock \showarticletitle{Possibilities determine the combinatorial structure
  of probability polytopes}.
\newblock \bibinfo{journal}{{\em J. Math. Psych.\/}}  \bibinfo{volume}{74}
  (\bibinfo{year}{2016}), \bibinfo{pages}{58--65}.
\newblock


\bibitem[\protect\citeauthoryear{Abramsky, Barbosa, and Mansfield}{Abramsky
  et~al\mbox{.}}{2017a}]%
        {abramsky2017contextual}
\bibfield{author}{\bibinfo{person}{S. Abramsky}, \bibinfo{person}{R.~S.
  Barbosa}, {and} \bibinfo{person}{S. Mansfield}.}
  \bibinfo{year}{2017}\natexlab{a}.
\newblock \showarticletitle{Contextual fraction as a measure of contextuality}.
\newblock \bibinfo{journal}{{\em Phys. Rev. Lett.\/}}  \bibinfo{volume}{119}
  (\bibinfo{year}{2017}), \bibinfo{pages}{050504}.
\newblock


\bibitem[\protect\citeauthoryear{Abramsky and Brandenburger}{Abramsky and
  Brandenburger}{2011}]%
        {AB}
\bibfield{author}{\bibinfo{person}{S. Abramsky} {and} \bibinfo{person}{A.
  Brandenburger}.} \bibinfo{year}{2011}\natexlab{}.
\newblock \showarticletitle{The sheaf-theoretic structure of non-locality and
  contextuality}.
\newblock \bibinfo{journal}{{\em New J. Phys.\/}} \bibinfo{volume}{13},
  \bibinfo{number}{11} (\bibinfo{year}{2011}), \bibinfo{pages}{113036}.
\newblock


\bibitem[\protect\citeauthoryear{Abramsky and Hardy}{Abramsky and
  Hardy}{2012}]%
        {logicalbi}
\bibfield{author}{\bibinfo{person}{S. Abramsky} {and} \bibinfo{person}{L.
  Hardy}.} \bibinfo{year}{2012}\natexlab{}.
\newblock \showarticletitle{Logical {B}ell inequalities}.
\newblock \bibinfo{journal}{{\em Phys. Rev. A\/}} \bibinfo{volume}{85},
  \bibinfo{number}{6} (\bibinfo{year}{2012}), \bibinfo{pages}{062114}.
\newblock


\bibitem[\protect\citeauthoryear{{Abramsky}, {S. Barbosa}, {Kishida}, {Lal},
  and {Mansfield}}{{Abramsky} et~al\mbox{.}}{2015}]%
        {ccp}
\bibfield{author}{\bibinfo{person}{S. {Abramsky}}, \bibinfo{person}{R. {S.
  Barbosa}}, \bibinfo{person}{K. {Kishida}}, \bibinfo{person}{R. {Lal}}, {and}
  \bibinfo{person}{S. {Mansfield}}.} \bibinfo{year}{2015}\natexlab{}.
\newblock \showarticletitle{{Contextuality, cohomology and paradox}}. In
  \bibinfo{booktitle}{{\em Comp. Sci. Logic (CSL) 2015}},
  Vol.~\bibinfo{volume}{41}. \bibinfo{pages}{211--228}.
\newblock


\bibitem[\protect\citeauthoryear{Ac{\'\i}n, Fritz, Leverrier, and
  Sainz}{Ac{\'\i}n et~al\mbox{.}}{2015}]%
        {acin}
\bibfield{author}{\bibinfo{person}{A. Ac{\'\i}n}, \bibinfo{person}{T. Fritz},
  \bibinfo{person}{A. Leverrier}, {and} \bibinfo{person}{A.~B. Sainz}.}
  \bibinfo{year}{2015}\natexlab{}.
\newblock \showarticletitle{A combinatorial approach to nonlocality and
  contextuality}.
\newblock \bibinfo{journal}{{\em Comm. Math. Phys.\/}} \bibinfo{volume}{334},
  \bibinfo{number}{2} (\bibinfo{year}{2015}), \bibinfo{pages}{533--628}.
\newblock


\bibitem[\protect\citeauthoryear{Anders and Browne}{Anders and Browne}{2009}]%
        {anders}
\bibfield{author}{\bibinfo{person}{J. Anders} {and} \bibinfo{person}{D.~E.
  Browne}.} \bibinfo{year}{2009}\natexlab{}.
\newblock \showarticletitle{Computational power of correlations}.
\newblock \bibinfo{journal}{{\em Phys. Rev. Lett.\/}} \bibinfo{volume}{102},
  \bibinfo{number}{5} (\bibinfo{year}{2009}), \bibinfo{pages}{050502}.
\newblock


\bibitem[\protect\citeauthoryear{Aspect, Grangier, and Roger}{Aspect
  et~al\mbox{.}}{1982}]%
        {aspect}
\bibfield{author}{\bibinfo{person}{A. Aspect}, \bibinfo{person}{P. Grangier},
  {and} \bibinfo{person}{G. Roger}.} \bibinfo{year}{1982}\natexlab{}.
\newblock \showarticletitle{Experimental Realization of
  {Einstein-Podolsky-Rosen-Bohm} Gedankenexperiment}.
\newblock \bibinfo{journal}{{\em Phys. Rev. Lett.\/}}  \bibinfo{volume}{49}
  (\bibinfo{year}{1982}), \bibinfo{pages}{91--94}.
\newblock


\bibitem[\protect\citeauthoryear{Atserias, Kolaitis, and Severini}{Atserias
  et~al\mbox{.}}{2017}]%
        {atserias2017generalized}
\bibfield{author}{\bibinfo{person}{A. Atserias}, \bibinfo{person}{P.~G.
  Kolaitis}, {and} \bibinfo{person}{S. Severini}.}
  \bibinfo{year}{2017}\natexlab{}.
\newblock \showarticletitle{Generalized satisfiability problems via operator
  assignments}. In \bibinfo{booktitle}{{\em Fund. Comp. Theory (FCT) 2017}},
  Vol.~\bibinfo{volume}{10472}. \bibinfo{pages}{56--68}.
\newblock


\bibitem[\protect\citeauthoryear{Barrett, Kent, and Pironio}{Barrett
  et~al\mbox{.}}{2006}]%
        {barrettmax}
\bibfield{author}{\bibinfo{person}{J. Barrett}, \bibinfo{person}{A. Kent},
  {and} \bibinfo{person}{S. Pironio}.} \bibinfo{year}{2006}\natexlab{}.
\newblock \showarticletitle{Maximally nonlocal and monogamous quantum
  correlations}.
\newblock \bibinfo{journal}{{\em Phys. Rev. Lett.\/}} \bibinfo{volume}{97},
  \bibinfo{number}{17} (\bibinfo{year}{2006}), \bibinfo{pages}{170409}.
\newblock


\bibitem[\protect\citeauthoryear{Bell}{Bell}{1964}]%
        {bell1964einstein}
\bibfield{author}{\bibinfo{person}{J.~S. Bell}.}
  \bibinfo{year}{1964}\natexlab{}.
\newblock \showarticletitle{On the {Einstein-Podolsky-Rosen} paradox}.
\newblock \bibinfo{journal}{{\em Physics\/}} \bibinfo{volume}{1},
  \bibinfo{number}{3} (\bibinfo{year}{1964}), \bibinfo{pages}{195--200}.
\newblock


\bibitem[\protect\citeauthoryear{Bell}{Bell}{1966}]%
        {bell66}
\bibfield{author}{\bibinfo{person}{J.~S. Bell}.}
  \bibinfo{year}{1966}\natexlab{}.
\newblock \showarticletitle{On the problem of hidden variables in quantum
  mechanics}.
\newblock \bibinfo{journal}{{\em Rev. Mod. Phys.\/}} \bibinfo{volume}{38},
  \bibinfo{number}{3} (\bibinfo{year}{1966}), \bibinfo{pages}{447}.
\newblock


\bibitem[\protect\citeauthoryear{Bravyi and Kitaev}{Bravyi and Kitaev}{2005}]%
        {bravkit}
\bibfield{author}{\bibinfo{person}{S. Bravyi} {and} \bibinfo{person}{A.
  Kitaev}.} \bibinfo{year}{2005}\natexlab{}.
\newblock \showarticletitle{Universal quantum computation with ideal {C}lifford
  gates and noisy ancillas}.
\newblock \bibinfo{journal}{{\em Phys. Rev. A\/}} \bibinfo{volume}{71},
  \bibinfo{number}{2} (\bibinfo{year}{2005}), \bibinfo{pages}{022316}.
\newblock


\bibitem[\protect\citeauthoryear{Bri{\"e}t, Buhrman, Lee, and Vidick}{Bri{\"e}t
  et~al\mbox{.}}{2013}]%
        {briet2013multipartite}
\bibfield{author}{\bibinfo{person}{J. Bri{\"e}t}, \bibinfo{person}{H. Buhrman},
  \bibinfo{person}{T. Lee}, {and} \bibinfo{person}{T. Vidick}.}
  \bibinfo{year}{2013}\natexlab{}.
\newblock \showarticletitle{Multipartite entanglement in XOR games}.
\newblock \bibinfo{journal}{{\em Quant. Inf. \& Comp.\/}} \bibinfo{volume}{13},
  \bibinfo{number}{3-4} (\bibinfo{year}{2013}), \bibinfo{pages}{334--360}.
\newblock


\bibitem[\protect\citeauthoryear{Cabello, Severini, and Winter}{Cabello
  et~al\mbox{.}}{2014}]%
        {csw}
\bibfield{author}{\bibinfo{person}{A. Cabello}, \bibinfo{person}{S. Severini},
  {and} \bibinfo{person}{A. Winter}.} \bibinfo{year}{2014}\natexlab{}.
\newblock \showarticletitle{Graph-theoretic approach to quantum correlations}.
\newblock \bibinfo{journal}{{\em Phys. Rev. Lett.\/}} \bibinfo{volume}{112},
  \bibinfo{number}{4} (\bibinfo{year}{2014}), \bibinfo{pages}{040401}.
\newblock


\bibitem[\protect\citeauthoryear{Campbell}{Campbell}{2014}]%
        {campbell2014enhanced}
\bibfield{author}{\bibinfo{person}{E.~T. Campbell}.}
  \bibinfo{year}{2014}\natexlab{}.
\newblock \showarticletitle{Enhanced fault-tolerant quantum computing in
  d-level systems}.
\newblock \bibinfo{journal}{{\em Phys. Rev. Lett.\/}} \bibinfo{volume}{113},
  \bibinfo{number}{23} (\bibinfo{year}{2014}), \bibinfo{pages}{230501}.
\newblock


\bibitem[\protect\citeauthoryear{Cleve, Hoyer, Toner, and Watrous}{Cleve
  et~al\mbox{.}}{2004}]%
        {cleve2004consequences}
\bibfield{author}{\bibinfo{person}{R. Cleve}, \bibinfo{person}{P. Hoyer},
  \bibinfo{person}{B. Toner}, {and} \bibinfo{person}{J. Watrous}.}
  \bibinfo{year}{2004}\natexlab{}.
\newblock \showarticletitle{Consequences and limits of nonlocal strategies}. In
  \bibinfo{booktitle}{{\em Proc. 19th IEEE Computational Complexity}}.
  \bibinfo{pages}{236--249}.
\newblock


\bibitem[\protect\citeauthoryear{Cleve and Mittal}{Cleve and Mittal}{2014}]%
        {cleve2014characterization}
\bibfield{author}{\bibinfo{person}{R. Cleve} {and} \bibinfo{person}{R.
  Mittal}.} \bibinfo{year}{2014}\natexlab{}.
\newblock \showarticletitle{Characterization of binary constraint system
  games}. In \bibinfo{booktitle}{{\em ICALP 2014}}. \bibinfo{pages}{320--331}.
\newblock


\bibitem[\protect\citeauthoryear{Cormick, Galv\~{a}o, Gottesman, Paz, and
  Pittenger}{Cormick et~al\mbox{.}}{2006}]%
        {cormick2006classicality}
\bibfield{author}{\bibinfo{person}{C. Cormick}, \bibinfo{person}{E.~F.
  Galv\~{a}o}, \bibinfo{person}{D. Gottesman}, \bibinfo{person}{J.~P. Paz},
  {and} \bibinfo{person}{A.~O. Pittenger}.} \bibinfo{year}{2006}\natexlab{}.
\newblock \showarticletitle{Classicality in discrete {W}igner functions}.
\newblock \bibinfo{journal}{{\em Phys. Rev. A\/}} \bibinfo{volume}{73},
  \bibinfo{number}{1} (\bibinfo{year}{2006}), \bibinfo{pages}{012301}.
\newblock


\bibitem[\protect\citeauthoryear{Cubitt, Leung, Matthews, and Winter}{Cubitt
  et~al\mbox{.}}{2010}]%
        {cubitt}
\bibfield{author}{\bibinfo{person}{T.~S. Cubitt}, \bibinfo{person}{D. Leung},
  \bibinfo{person}{W. Matthews}, {and} \bibinfo{person}{A. Winter}.}
  \bibinfo{year}{2010}\natexlab{}.
\newblock \showarticletitle{Improving zero-error classical communication with
  entanglement}.
\newblock \bibinfo{journal}{{\em Phys. Rev. Lett.\/}} \bibinfo{volume}{104},
  \bibinfo{number}{23} (\bibinfo{year}{2010}), \bibinfo{pages}{230503}.
\newblock


\bibitem[\protect\citeauthoryear{Cui, Gottesman, and Krishna}{Cui
  et~al\mbox{.}}{2017}]%
        {cui2017diagonal}
\bibfield{author}{\bibinfo{person}{S.~X. Cui}, \bibinfo{person}{D. Gottesman},
  {and} \bibinfo{person}{A. Krishna}.} \bibinfo{year}{2017}\natexlab{}.
\newblock \showarticletitle{Diagonal gates in the Clifford hierarchy}.
\newblock \bibinfo{journal}{{\em Phys. Rev. A\/}} \bibinfo{volume}{95},
  \bibinfo{number}{1} (\bibinfo{year}{2017}), \bibinfo{pages}{012329}.
\newblock


\bibitem[\protect\citeauthoryear{de~Silva}{de~Silva}{2017}]%
        {de2017graph}
\bibfield{author}{\bibinfo{person}{N. de Silva}.}
  \bibinfo{year}{2017}\natexlab{}.
\newblock \showarticletitle{Graph-theoretic strengths of contextuality}.
\newblock \bibinfo{journal}{{\em Phys. Rev. A\/}} \bibinfo{volume}{95},
  \bibinfo{number}{3} (\bibinfo{year}{2017}), \bibinfo{pages}{032108}.
\newblock


\bibitem[\protect\citeauthoryear{Delfosse, Allard~Guerin, Bian, and
  Raussendorf}{Delfosse et~al\mbox{.}}{2015}]%
        {delfosse2015wigner}
\bibfield{author}{\bibinfo{person}{N. Delfosse}, \bibinfo{person}{P.
  Allard~Guerin}, \bibinfo{person}{J. Bian}, {and} \bibinfo{person}{R.
  Raussendorf}.} \bibinfo{year}{2015}\natexlab{}.
\newblock \showarticletitle{Wigner function negativity and contextuality in
  quantum computation on rebits}.
\newblock \bibinfo{journal}{{\em Phys. Rev. X\/}} \bibinfo{volume}{5},
  \bibinfo{number}{2} (\bibinfo{year}{2015}), \bibinfo{pages}{021003}.
\newblock


\bibitem[\protect\citeauthoryear{Delfosse, Okay, Bermejo-Vega, Browne, and
  Raussendorf}{Delfosse et~al\mbox{.}}{2017}]%
        {delfosse2016}
\bibfield{author}{\bibinfo{person}{N. Delfosse}, \bibinfo{person}{C. Okay},
  \bibinfo{person}{J. Bermejo-Vega}, \bibinfo{person}{D.~E. Browne}, {and}
  \bibinfo{person}{R. Raussendorf}.} \bibinfo{year}{2017}\natexlab{}.
\newblock \showarticletitle{Equivalence between contextuality and negativity of
  the Wigner function for qudits}.
\newblock \bibinfo{journal}{{\em New J. Phys.\/}} \bibinfo{volume}{19},
  \bibinfo{number}{12} (\bibinfo{year}{2017}), \bibinfo{pages}{123024}.
\newblock


\bibitem[\protect\citeauthoryear{Dickson}{Dickson}{1896}]%
        {dickson1896analytic}
\bibfield{author}{\bibinfo{person}{L.~E. Dickson}.}
  \bibinfo{year}{1896}\natexlab{}.
\newblock \showarticletitle{The analytic representation of substitutions on a
  power of a prime number of letters with a discussion of the linear group}.
\newblock \bibinfo{journal}{{\em Ann. Math.\/}}  \bibinfo{volume}{11}
  (\bibinfo{year}{1896}), \bibinfo{pages}{65--120}.
\newblock


\bibitem[\protect\citeauthoryear{Einstein, Podolsky, and Rosen}{Einstein
  et~al\mbox{.}}{1935}]%
        {EPR}
\bibfield{author}{\bibinfo{person}{A. Einstein}, \bibinfo{person}{B. Podolsky},
  {and} \bibinfo{person}{N. Rosen}.} \bibinfo{year}{1935}\natexlab{}.
\newblock \showarticletitle{Can quantum-mechanical description of physical
  reality be considered complete?}
\newblock \bibinfo{journal}{{\em Phys. Rev.\/}}  \bibinfo{volume}{47}
  (\bibinfo{year}{1935}), \bibinfo{pages}{777--780}.
\newblock


\bibitem[\protect\citeauthoryear{Elitzur, Popescu, and Rohrlich}{Elitzur
  et~al\mbox{.}}{1992}]%
        {eprohr}
\bibfield{author}{\bibinfo{person}{A.~C. Elitzur}, \bibinfo{person}{S.
  Popescu}, {and} \bibinfo{person}{D. Rohrlich}.}
  \bibinfo{year}{1992}\natexlab{}.
\newblock \showarticletitle{Quantum nonlocality for each pair in an ensemble}.
\newblock \bibinfo{journal}{{\em Phys. Lett. A\/}} \bibinfo{volume}{162},
  \bibinfo{number}{1} (\bibinfo{year}{1992}), \bibinfo{pages}{25--28}.
\newblock


\bibitem[\protect\citeauthoryear{Fine}{Fine}{1982}]%
        {fine}
\bibfield{author}{\bibinfo{person}{A. Fine}.} \bibinfo{year}{1982}\natexlab{}.
\newblock \showarticletitle{Hidden variables, joint probability, and the Bell
  inequalities}.
\newblock \bibinfo{journal}{{\em Phys. Rev. Lett.\/}} \bibinfo{volume}{48},
  \bibinfo{number}{5} (\bibinfo{year}{1982}), \bibinfo{pages}{291}.
\newblock


\bibitem[\protect\citeauthoryear{Galv\~ao}{Galv\~ao}{2005}]%
        {galvao2005discrete}
\bibfield{author}{\bibinfo{person}{E.~F. Galv\~ao}.}
  \bibinfo{year}{2005}\natexlab{}.
\newblock \showarticletitle{Discrete Wigner functions and quantum computational
  speedup}.
\newblock \bibinfo{journal}{{\em Phys. Rev. A\/}} \bibinfo{volume}{71},
  \bibinfo{number}{4} (\bibinfo{year}{2005}), \bibinfo{pages}{042302}.
\newblock


\bibitem[\protect\citeauthoryear{Gibbons, Hoffman, and Wootters}{Gibbons
  et~al\mbox{.}}{2004}]%
        {gibbons2004discrete}
\bibfield{author}{\bibinfo{person}{K.~S. Gibbons}, \bibinfo{person}{M.~J.
  Hoffman}, {and} \bibinfo{person}{W.~K. Wootters}.}
  \bibinfo{year}{2004}\natexlab{}.
\newblock \showarticletitle{Discrete phase space based on finite fields}.
\newblock \bibinfo{journal}{{\em Phys. Rev. A\/}} \bibinfo{volume}{70},
  \bibinfo{number}{6} (\bibinfo{year}{2004}), \bibinfo{pages}{062101}.
\newblock


\bibitem[\protect\citeauthoryear{Gottesman}{Gottesman}{1997}]%
        {gotthesis}
\bibfield{author}{\bibinfo{person}{D. Gottesman}.}
  \bibinfo{year}{1997}\natexlab{}.
\newblock {\em \bibinfo{title}{Stabilizer codes and quantum error correction}}.
\newblock \bibinfo{thesistype}{Ph.D. Dissertation}. \bibinfo{school}{Caltech}.
\newblock


\bibitem[\protect\citeauthoryear{Gottesman}{Gottesman}{1998}]%
        {gk}
\bibfield{author}{\bibinfo{person}{D. Gottesman}.}
  \bibinfo{year}{1998}\natexlab{}.
\newblock \showarticletitle{Theory of fault-tolerant quantum computation}.
\newblock \bibinfo{journal}{{\em Phys. Rev. A\/}}  \bibinfo{volume}{57}
  (\bibinfo{year}{1998}), \bibinfo{pages}{127--137}.
\newblock


\bibitem[\protect\citeauthoryear{Gottesman and Chuang}{Gottesman and
  Chuang}{1999}]%
        {gottesman1999demonstrating}
\bibfield{author}{\bibinfo{person}{D. Gottesman} {and} \bibinfo{person}{I.
  Chuang}.} \bibinfo{year}{1999}\natexlab{}.
\newblock \showarticletitle{Demonstrating the viability of universal quantum
  computation using teleportation and single-qubit operations}.
\newblock \bibinfo{journal}{{\em Nature\/}} \bibinfo{volume}{402},
  \bibinfo{number}{6760} (\bibinfo{year}{1999}), \bibinfo{pages}{390--393}.
\newblock


\bibitem[\protect\citeauthoryear{Greenberger, Horne, Shimony, and
  Zeilinger}{Greenberger et~al\mbox{.}}{1990}]%
        {ghz}
\bibfield{author}{\bibinfo{person}{D.~M. Greenberger}, \bibinfo{person}{M.~A.
  Horne}, \bibinfo{person}{A. Shimony}, {and} \bibinfo{person}{A. Zeilinger}.}
  \bibinfo{year}{1990}\natexlab{}.
\newblock \showarticletitle{Bell's theorem without inequalities}.
\newblock \bibinfo{journal}{{\em Am. J. Phys.\/}} \bibinfo{volume}{58},
  \bibinfo{number}{12} (\bibinfo{year}{1990}), \bibinfo{pages}{1131--1143}.
\newblock


\bibitem[\protect\citeauthoryear{Gross}{Gross}{2006}]%
        {gross2006hudson}
\bibfield{author}{\bibinfo{person}{D. Gross}.} \bibinfo{year}{2006}\natexlab{}.
\newblock \showarticletitle{Hudson's theorem for finite-dimensional quantum
  systems}.
\newblock \bibinfo{journal}{{\em J. Math. Phys.\/}} \bibinfo{volume}{47},
  \bibinfo{number}{12} (\bibinfo{year}{2006}), \bibinfo{pages}{122107}.
\newblock


\bibitem[\protect\citeauthoryear{Grudka, Horodecki, Horodecki, Horodecki,
  Horodecki, Joshi, K{\l}obus, and W{\'o}jcik}{Grudka et~al\mbox{.}}{2014}]%
        {grudka}
\bibfield{author}{\bibinfo{person}{A. Grudka}, \bibinfo{person}{K. Horodecki},
  \bibinfo{person}{M. Horodecki}, \bibinfo{person}{P. Horodecki},
  \bibinfo{person}{R. Horodecki}, \bibinfo{person}{P. Joshi},
  \bibinfo{person}{W. K{\l}obus}, {and} \bibinfo{person}{A. W{\'o}jcik}.}
  \bibinfo{year}{2014}\natexlab{}.
\newblock \showarticletitle{Quantifying contextuality}.
\newblock \bibinfo{journal}{{\em Phys. Rev. Lett.\/}}  \bibinfo{volume}{112}
  (\bibinfo{year}{2014}), \bibinfo{pages}{120401}.
\newblock


\bibitem[\protect\citeauthoryear{Hensen et~al\mbox{.}}{Hensen
  et~al\mbox{.}}{2015}]%
        {hensen2015loophole}
\bibfield{author}{\bibinfo{person}{B. Hensen} {et~al\mbox{.}}}
  \bibinfo{year}{2015}\natexlab{}.
\newblock \showarticletitle{Loophole-free Bell inequality violation using
  electron spins separated by 1.3 kilometres}.
\newblock \bibinfo{journal}{{\em Nature\/}} \bibinfo{volume}{526},
  \bibinfo{number}{7575} (\bibinfo{year}{2015}), \bibinfo{pages}{682--686}.
\newblock


\bibitem[\protect\citeauthoryear{Hou}{Hou}{2015}]%
        {hou2015permutation}
\bibfield{author}{\bibinfo{person}{X. Hou}.} \bibinfo{year}{2015}\natexlab{}.
\newblock \showarticletitle{Permutation polynomials over finite fields---a
  survey of recent advances}.
\newblock \bibinfo{journal}{{\em Finite Fields and Their Applications\/}}
  \bibinfo{volume}{32} (\bibinfo{year}{2015}), \bibinfo{pages}{82--119}.
\newblock


\bibitem[\protect\citeauthoryear{Howard and Campbell}{Howard and
  Campbell}{2017}]%
        {howard2017application}
\bibfield{author}{\bibinfo{person}{M. Howard} {and} \bibinfo{person}{E.
  Campbell}.} \bibinfo{year}{2017}\natexlab{}.
\newblock \showarticletitle{Application of a resource theory for magic states
  to fault-tolerant quantum computing}.
\newblock \bibinfo{journal}{{\em Phys. Rev. Lett.\/}} \bibinfo{volume}{118},
  \bibinfo{number}{9} (\bibinfo{year}{2017}), \bibinfo{pages}{090501}.
\newblock


\bibitem[\protect\citeauthoryear{Howard and Vala}{Howard and Vala}{2012}]%
        {howard2012qudit}
\bibfield{author}{\bibinfo{person}{M. Howard} {and} \bibinfo{person}{J. Vala}.}
  \bibinfo{year}{2012}\natexlab{}.
\newblock \showarticletitle{Qudit versions of the qubit $\pi$/8 gate}.
\newblock \bibinfo{journal}{{\em Phys. Rev. A\/}} \bibinfo{volume}{86},
  \bibinfo{number}{2} (\bibinfo{year}{2012}), \bibinfo{pages}{022316}.
\newblock


\bibitem[\protect\citeauthoryear{Howard, Wallman, Veitch, and Emerson}{Howard
  et~al\mbox{.}}{2014}]%
        {howard}
\bibfield{author}{\bibinfo{person}{M. Howard}, \bibinfo{person}{J. Wallman},
  \bibinfo{person}{V. Veitch}, {and} \bibinfo{person}{J. Emerson}.}
  \bibinfo{year}{2014}\natexlab{}.
\newblock \showarticletitle{Contextuality supplies the `magic' for quantum
  computation}.
\newblock \bibinfo{journal}{{\em Nature\/}} \bibinfo{volume}{510},
  \bibinfo{number}{7505} (\bibinfo{year}{2014}), \bibinfo{pages}{351--355}.
\newblock


\bibitem[\protect\citeauthoryear{Kishida}{Kishida}{2016}]%
        {Kishida16}
\bibfield{author}{\bibinfo{person}{K. Kishida}.}
  \bibinfo{year}{2016}\natexlab{}.
\newblock \showarticletitle{Logic of local inference for contextuality in
  quantum physics and beyond}. In \bibinfo{booktitle}{{\em ICALP 2016}},
  Vol.~\bibinfo{volume}{55}. \bibinfo{pages}{113:1--113:14}.
\newblock


\bibitem[\protect\citeauthoryear{Kleinmann, G{\"u}hne, Portillo, Larsson, and
  Cabello}{Kleinmann et~al\mbox{.}}{2011}]%
        {memory}
\bibfield{author}{\bibinfo{person}{M. Kleinmann}, \bibinfo{person}{O.
  G{\"u}hne}, \bibinfo{person}{J.~R. Portillo}, \bibinfo{person}{J. Larsson},
  {and} \bibinfo{person}{A. Cabello}.} \bibinfo{year}{2011}\natexlab{}.
\newblock \showarticletitle{Memory cost of quantum contextuality}.
\newblock \bibinfo{journal}{{\em New J. Phys.\/}} \bibinfo{volume}{13},
  \bibinfo{number}{11} (\bibinfo{year}{2011}), \bibinfo{pages}{113011}.
\newblock


\bibitem[\protect\citeauthoryear{Klyachko, Can, Binicio\u{g}lu, and
  Shumovsky}{Klyachko et~al\mbox{.}}{2008}]%
        {klyachko}
\bibfield{author}{\bibinfo{person}{A.~A. Klyachko}, \bibinfo{person}{M.~Ali
  Can}, \bibinfo{person}{S. Binicio\u{g}lu}, {and} \bibinfo{person}{A.~S.
  Shumovsky}.} \bibinfo{year}{2008}\natexlab{}.
\newblock \showarticletitle{A simple test for hidden variables in spin-1
  system}.
\newblock \bibinfo{journal}{{\em Phys. Rev. Lett.\/}} \bibinfo{volume}{101},
  \bibinfo{number}{2} (\bibinfo{year}{2008}), \bibinfo{pages}{020403}.
\newblock


\bibitem[\protect\citeauthoryear{Kochen and Specker}{Kochen and
  Specker}{1967}]%
        {ks}
\bibfield{author}{\bibinfo{person}{S. Kochen} {and} \bibinfo{person}{E.~P.
  Specker}.} \bibinfo{year}{1967}\natexlab{}.
\newblock \showarticletitle{The problem of hidden variables in quantum
  mechanics}.
\newblock \bibinfo{journal}{{\em J. Math. Mech.\/}}  \bibinfo{volume}{17}
  (\bibinfo{year}{1967}), \bibinfo{pages}{59--87}.
\newblock


\bibitem[\protect\citeauthoryear{Mari and Eisert}{Mari and Eisert}{2012}]%
        {mari2012positive}
\bibfield{author}{\bibinfo{person}{A. Mari} {and} \bibinfo{person}{J. Eisert}.}
  \bibinfo{year}{2012}\natexlab{}.
\newblock \showarticletitle{Positive {W}igner functions render classical
  simulation of quantum computation efficient}.
\newblock \bibinfo{journal}{{\em Phys. Rev. Lett.\/}} \bibinfo{volume}{109},
  \bibinfo{number}{23} (\bibinfo{year}{2012}), \bibinfo{pages}{230503}.
\newblock


\bibitem[\protect\citeauthoryear{Mermin}{Mermin}{1993}]%
        {mermin}
\bibfield{author}{\bibinfo{person}{N.~D. Mermin}.}
  \bibinfo{year}{1993}\natexlab{}.
\newblock \showarticletitle{Hidden variables and the two theorems of {John
  Bell}}.
\newblock \bibinfo{journal}{{\em Rev. Mod. Phys.\/}} \bibinfo{volume}{65},
  \bibinfo{number}{3} (\bibinfo{year}{1993}), \bibinfo{pages}{803}.
\newblock


\bibitem[\protect\citeauthoryear{Okay, Roberts, Bartlett, and Raussendorf}{Okay
  et~al\mbox{.}}{2017}]%
        {okay2017topological}
\bibfield{author}{\bibinfo{person}{C. Okay}, \bibinfo{person}{S. Roberts},
  \bibinfo{person}{S.~D. Bartlett}, {and} \bibinfo{person}{R. Raussendorf}.}
  \bibinfo{year}{2017}\natexlab{}.
\newblock \showarticletitle{Topological proofs of contextuality in quantum
  mechanics}.
\newblock \bibinfo{journal}{{\em Quant. Inf. \& Comp.\/}} \bibinfo{volume}{17},
  \bibinfo{number}{13-14} (\bibinfo{year}{2017}), \bibinfo{pages}{1135--1166}.
\newblock


\bibitem[\protect\citeauthoryear{Pashayan, Wallman, and Bartlett}{Pashayan
  et~al\mbox{.}}{2015}]%
        {pashayan2015estimating}
\bibfield{author}{\bibinfo{person}{H. Pashayan}, \bibinfo{person}{J.~J.
  Wallman}, {and} \bibinfo{person}{S.~D. Bartlett}.}
  \bibinfo{year}{2015}\natexlab{}.
\newblock \showarticletitle{Estimating outcome probabilities of quantum
  circuits using quasiprobabilities}.
\newblock \bibinfo{journal}{{\em Phys. Rev. Lett.\/}} \bibinfo{volume}{115},
  \bibinfo{number}{7} (\bibinfo{year}{2015}), \bibinfo{pages}{070501}.
\newblock


\bibitem[\protect\citeauthoryear{Peres}{Peres}{1999}]%
        {peres1999all}
\bibfield{author}{\bibinfo{person}{A. Peres}.} \bibinfo{year}{1999}\natexlab{}.
\newblock \showarticletitle{All the Bell inequalities}.
\newblock \bibinfo{journal}{{\em Found. Phys.\/}} \bibinfo{volume}{29},
  \bibinfo{number}{4} (\bibinfo{year}{1999}), \bibinfo{pages}{589--614}.
\newblock


\bibitem[\protect\citeauthoryear{Popescu and Rohrlich}{Popescu and
  Rohrlich}{1994}]%
        {popescu1994quantum}
\bibfield{author}{\bibinfo{person}{S. Popescu} {and} \bibinfo{person}{D.
  Rohrlich}.} \bibinfo{year}{1994}\natexlab{}.
\newblock \showarticletitle{Quantum nonlocality as an axiom}.
\newblock \bibinfo{journal}{{\em Found. Phys.\/}} \bibinfo{volume}{24},
  \bibinfo{number}{3} (\bibinfo{year}{1994}), \bibinfo{pages}{379--385}.
\newblock


\bibitem[\protect\citeauthoryear{Raussendorf}{Raussendorf}{2013}]%
        {raussendorf}
\bibfield{author}{\bibinfo{person}{R. Raussendorf}.}
  \bibinfo{year}{2013}\natexlab{}.
\newblock \showarticletitle{Contextuality in measurement-based quantum
  computation}.
\newblock \bibinfo{journal}{{\em Phys. Rev. A\/}} \bibinfo{volume}{88},
  \bibinfo{number}{2} (\bibinfo{year}{2013}), \bibinfo{pages}{022322}.
\newblock


\bibitem[\protect\citeauthoryear{Raussendorf and Briegel}{Raussendorf and
  Briegel}{2001}]%
        {raussendorf2001one}
\bibfield{author}{\bibinfo{person}{R. Raussendorf} {and} \bibinfo{person}{H.~J.
  Briegel}.} \bibinfo{year}{2001}\natexlab{}.
\newblock \showarticletitle{A one-way quantum computer}.
\newblock \bibinfo{journal}{{\em Phys. Rev. Lett.\/}} \bibinfo{volume}{86},
  \bibinfo{number}{22} (\bibinfo{year}{2001}), \bibinfo{pages}{5188}.
\newblock


\bibitem[\protect\citeauthoryear{van Dam}{van Dam}{1999}]%
        {vandamthesis}
\bibfield{author}{\bibinfo{person}{W. van Dam}.}
  \bibinfo{year}{1999}\natexlab{}.
\newblock {\em \bibinfo{title}{Nonlocality \& communication complexity}}.
\newblock \bibinfo{thesistype}{Ph.D. Dissertation}.
\newblock


\bibitem[\protect\citeauthoryear{Veitch, Ferrie, Gross, and Emerson}{Veitch
  et~al\mbox{.}}{2012}]%
        {veitch}
\bibfield{author}{\bibinfo{person}{V. Veitch}, \bibinfo{person}{C. Ferrie},
  \bibinfo{person}{D. Gross}, {and} \bibinfo{person}{J. Emerson}.}
  \bibinfo{year}{2012}\natexlab{}.
\newblock \showarticletitle{Negative quasi-probability as a resource for
  quantum computation}.
\newblock \bibinfo{journal}{{\em New J. Phys.\/}} \bibinfo{volume}{14},
  \bibinfo{number}{11} (\bibinfo{year}{2012}), \bibinfo{pages}{113011}.
\newblock


\bibitem[\protect\citeauthoryear{Veitch, Mousavian, Gottesman, and
  Emerson}{Veitch et~al\mbox{.}}{2014}]%
        {veitch2}
\bibfield{author}{\bibinfo{person}{V. Veitch}, \bibinfo{person}{H. Mousavian},
  \bibinfo{person}{D. Gottesman}, {and} \bibinfo{person}{J. Emerson}.}
  \bibinfo{year}{2014}\natexlab{}.
\newblock \showarticletitle{The resource theory of stabilizer quantum
  computation}.
\newblock \bibinfo{journal}{{\em New J. Phys.\/}} \bibinfo{volume}{16},
  \bibinfo{number}{1} (\bibinfo{year}{2014}), \bibinfo{pages}{013009}.
\newblock


\bibitem[\protect\citeauthoryear{Wootters}{Wootters}{1986}]%
        {wootters1986discrete}
\bibfield{author}{\bibinfo{person}{W.~K. Wootters}.}
  \bibinfo{year}{1986}\natexlab{}.
\newblock \showarticletitle{The discrete Wigner function}.
\newblock \bibinfo{journal}{{\em Ann. N. Y. Acad. Sci.\/}}
  \bibinfo{volume}{480}, \bibinfo{number}{1} (\bibinfo{year}{1986}),
  \bibinfo{pages}{275--282}.
\newblock


\end{thebibliography}

\end{document}